\documentclass{eptcs}

\providecommand{\event}{WRS 2011} 

\providecommand{\firstpage}{1}

\usepackage[utf8x]{inputenc}
\usepackage{amsmath}
\usepackage{amsfonts}
\usepackage{amssymb}
\usepackage{amsthm}
\usepackage{microtype}

\DeclareSymbolFont{symbols}{OMS}{ptm}{m}{n}

\usepackage[bf]{caption}

\usepackage{tikz}
\usetikzlibrary{arrows,decorations.pathmorphing,backgrounds,positioning,fit,%
shapes,shapes.gates.logic.US}

\hypersetup{
  pdfborder = 0 0 0.5,
  pdfdisplaydoctitle,
  pdftitle={Productivity of Non-Orthogonal Term Rewrite Systems},
  pdfauthor={Matthias Raffelsieper},
}

\begin{document}
\renewcommand{\email}[1]{\href{mailto:#1}{\footnotesize\tt #1}}

\title{Productivity of Non-Orthogonal Term Rewrite Systems}
\author{Matthias Raffelsieper
\institute{
Department of Computer Science, TU Eindhoven\\
P.O.\ Box 513, 5600~MB Eindhoven, The Netherlands\\
eMail: \email{M.Raffelsieper@tue.nl}
}
}
\def\titlerunning{Productivity of Non-Orthogonal Term Rewrite Systems}
\def\authorrunning{Matthias Raffelsieper}

\newcommand{\Nats}{\mathbb{N}}
\newcommand{\rt}{\mathrm{root}}
\newcommand{\T}{\mathcal{T}}
\newcommand{\V}{\mathcal{V}}
\newcommand{\C}{\mathcal{C}}
\newcommand{\R}{\mathcal{R}}
\newcommand{\NF}{\mathrm{NF}}
\newcommand{\gnd}{\mathrm{gnd}}
\newcommand{\Spec}{\mathcal{S}}
\newcommand{\ar}{\mathsf{ar}}
\newcommand{\supt}{\trianglerighteq}
\newcommand{\subt}{\trianglelefteq}
\newcommand{\fsym}[1]{\mathsf{#1}}
\newcommand{\zero}{\fsym{0}}
\newcommand{\one}{\fsym{1}}
\newcommand{\Pos}{\mathrm{Pos}}
\newcommand{\PosRedS}{\mathrm{Pos}^{\mathrm{red}_s}}
\newcommand{\blocked}[1][\mu]{\mathrm{blocked}_{#1}}
\newcommand{\parto}{\mathrel{\smash{\overset{\shortparallel}{\rightarrow}}}}
\newcommand{\muto}{\mathrel{\smash{\overset{\mu}{\rightarrow}}}}
\newcommand{\SNinf}{\mathsf{SN}^{\infty}}

\def\clap#1{\hbox to 0pt{\hss#1\hss}}
\def\mathllap{\mathpalette\mathllapinternal}
\def\mathrlap{\mathpalette\mathrlapinternal}
\def\mathclap{\mathpalette\mathclapinternal}
\def\mathllapinternal#1#2{%
\llap{$\mathsurround=0pt#1{#2}$}}
\def\mathrlapinternal#1#2{%
\rlap{$\mathsurround=0pt#1{#2}$}}
\def\mathclapinternal#1#2{%
\clap{$\mathsurround=0pt#1{#2}$}}

\theoremstyle{plain}
\newtheorem{theorem}{Theorem}
\newtheorem{lemma}[theorem]{Lemma}
\newtheorem{proposition}[theorem]{Proposition}

\theoremstyle{definition}
\newtheorem{definition}[theorem]{Definition}
\newtheorem{example}[theorem]{Example}

\maketitle

\begin{abstract}
Productivity is the property that finite prefixes of an infinite constructor
term can be computed using a given term rewrite system. Hitherto, productivity
has only been considered for orthogonal systems, where non-determinism is not
allowed. This paper presents techniques to also prove productivity of
non-orthogonal term rewrite systems.
For such systems, it is desired that one does not have to guess the reduction
steps to perform, instead any outermost-fair reduction should compute
an infinite constructor term in the limit.
As a main result, it is shown that for possibly non-orthogonal term rewrite
systems this kind of productivity can be concluded from context-sensitive
termination.
This result can be applied to prove stabilization of digital circuits, as
will be illustrated by means of an example.
\end{abstract}

\section{Introduction}
\label{sec:Introduction}

Productivity is the property that a given set of computation rules computes a
desired infinite object. This has been studied mostly in the setting of
\emph{streams}, the simplest infinite objects. However, as already observed
in~\cite{ZR10}, productivity is also of interest for other infinite structures,
for example infinite trees, or mixtures of finite and infinite structures.
A prominent example of the latter are lists in the programming language
\textsf{Haskell}~\cite{Haskell98}, which can be finite (by ending with a
sentinel ``\texttt{[]}'') or which can go on forever.

Existing approaches for automatically checking productivity,
e.g.,~\cite{End10,EGH08,ZR10},
are restricted to \emph{orthogonal} systems. The main reason for this
restriction is that it disallows non-determinism. A complete computer program
(i.e., a program and all possible input sequences, neglecting sources of true
randomness) always behaves deterministically, as the steps of computation are
precisely determined.
However, often a complete program is not available, too large to be studied, or
its inputs are provided by the user or they are not specified completely.
In this case, non-determinism can be used to abstract from certain parts by
describing a number of possible behaviors. In such a setting, the restriction to
orthogonal systems, which is even far stronger than only disallowing
non-determinism, should be removed.
An example of such a setting are hardware components, describing streams of
output values which are depending on the streams of input values. To analyze
such components in isolation, all possible input streams have to be considered.

This paper presents an extension of the techniques in~\cite{ZR10} to analyze
productivity of specifications that may contain non-determinism.
As in that work, the main technique to prove productivity is by analyzing
termination of a corresponding context-sensitive term rewrite system~\cite{L98}.
Here however, overlapping rules are allowed and the data TRS is only required to
be terminating, but it need not be confluent nor left-linear.
This technique can be used to prove stabilization of hardware circuits,
which have external inputs whose exact sequence of values is unknown.
Thus, stabilization should be proven for all possible input sequences, which are
therefore abstracted to be random Boolean streams, i.e., arbitrary streams
containing the data values~$\zero$ and~$\one$.

\paragraph{Structure of the Paper.} In Section~\ref{sec:Specifications} we
introduce \emph{proper specifications}, which are the forms of rewrite systems
studied in this paper. After that, in Section~\ref{sec:Productivity}, the
different notions of productivity are discussed. For non-orthogonal
specifications as studied in this paper, there exist both \emph{weak} and
\emph{strong} productivity. We will motivate that strong productivity is the
notion that we are interested in, as it does guarantee a constructor term to
be reached by any outermost-fair reduction.
The theoretical basis is laid in Section~\ref{sec:Criteria},
proving our desired result that termination of a corresponding context-sensitive
TRS implies strong productivity of a proper specification.
Section~\ref{sec:ApplHW} then applies this theory to an example hardware
circuit, checking that for a given circuit the output values always stabilize,
regardless of the sequence of input values.
Finally, Section~\ref{sec:Conclusions} concludes the paper.

\section{Specifications}
\label{sec:Specifications}

A \emph{specification} gives the symbols and rules that shall be used to
compute an intended infinite object. This section gives a brief introduction to
term rewriting, mainly aimed at fixing notation. For an in-depth description of
term rewriting, see for example~\cite{BN98, Terese03}.
All symbols are assumed to have one of two possible \emph{sorts}.
The first sort~$d$ is for \emph{data}. Terms of this sort represent the
elements in an infinite structure, but which are not infinite terms by
themselves. An example for data are the Booleans $\fsym{false}$ and
$\fsym{true}$ (which are also written~$\zero$ and~$\one$),
or the natural numbers represented in Peano form by the two constructors $\zero$
and $\fsym{succ}$. The set of all terms of sort~$d$ is denoted
$\T_d(\Sigma_d, \V_d)$, where $\Sigma_d$ is a set of function symbols all having
types of the form $d^m \to d$ and where $\V_d$ is a set of \emph{variables}
all having sort~$d$.
The second sort is the sort~$s$ for \emph{structure}. Terms of this sort are to
represent the intended structure containing the data and therefore are allowed
to be infinite. The set of all well-typed structure terms is denoted
$\T_s(\Sigma_d \cup \Sigma_s, \V)$, where
$\Sigma_s$ is disjoint from $\Sigma_d$ and contains function symbols having
types of the form $d^m \times s^n \to s$ and where
$\V = \V_d \cup \V_s$ for a set $\V_s$ of variables all having sort~$s$, which
is disjoint from $\V_d$.
We define the set of all well-typed terms as
$\T(\Sigma_d \cup \Sigma_s, \V) = \T_d(\Sigma_d, \V_d) \cup
\T_s(\Sigma_d \cup \Sigma_s, \V)$ and denote
the set of all \emph{ground terms}, i.e., terms not containing any variables,
by $\T(\Sigma_d \cup \Sigma_s) =
\T(\Sigma_d \cup \Sigma_s, \emptyset)$.
A term $t \in \T(\Sigma_d \cup \Sigma_s, \V)$ of sort $\varsigma \in \{ d, s \}$
is either a variable, i.e., $t \in \V_{\varsigma}$,
or $t = f(u_1, \dotsc, u_m, t_1, \dotsc, t_n)$
with $f \in \Sigma_{\varsigma}$ of type $d^m \times s^n \to \varsigma$
(where $n=0$ if $\varsigma=d$),
$u_1, \dotsc, u_m \in \T_d(\Sigma_d, \V_d)$,
and $t_1, \dotsc, t_n \in \T_s(\Sigma_d \cup \Sigma_s, \V)$.
In the latter case, i.e., when $t = f(u_1, \dotsc, u_m, t_1, \dotsc, t_n)$, we
define the \emph{root} of the term~$t$ as $\rt(t) = f$.

A \emph{Term Rewrite System (TRS)} over a signature $\Sigma$ is a collection of
rules $(\ell,r) \in \T(\Sigma,\V)^2$ such that $\ell \notin \V$ and every
variable contained in~$r$ is also contained in~$\ell$. As usual, we write
$\ell \to r$ instead of $(\ell,r)$.
A term $t \in \T(\Sigma,\V)$ \emph{rewrites} to a term $t' \in \T(\Sigma,\V)$
with the rule $\ell \to r \in \R$, denoted $t \to_{\ell \to r,p} t'$
at \emph{position} $p \in \Pos(t)$, if a substitution~$\sigma$
exists such that $t|_p = \ell\sigma$ and $t' = t[r\sigma]_p$.
A position is as usual a sequence of natural numbers that identifies a number of
argument positions taken to reach a certain subterm. The notation $t[r\sigma]_p$
represents the term $t$ in which the subterm at position~$p$, that is denoted
by $t|_p$, has been replaced by the term $r\sigma$. This is the term~$r$ in
which all variables have been replaced according to the substitution $\sigma$,
which is a map from variables to terms.
It is allowed to only indicate the term rewrite system $\R$ instead of
the specific rule $\ell \to r$ or to leave out the subscripts in case they
are irrelevant or clear from the context.
The set of all \emph{normal forms} of a TRS $\R$ over a
signature $\Sigma$ is denoted $\NF(\R)$ and is defined as $\NF(\R) =
\{ t \in \T(\Sigma,\V) \mid \forall t' \in \T(\Sigma,\V): t \not\to_{\R} t' \}$.
The set of \emph{ground normal forms} $\NF_{\gnd}(\R)$ additionally requires
that all contained terms are ground terms, i.e.,
$\NF_{\gnd}(\R) = \NF(\R) \cap \T(\Sigma)$.

We still have to impose some restrictions on specifications to make our
approach work. These restrictions are given below in the definition of
\emph{proper} specifications, which are similar to those of~\cite{ZR10}.

\begin{definition}
\label{def:Spec}
A \emph{proper specification} is a tuple
$\Spec = (\Sigma_d,\Sigma_s,\C,\R_d,\R_s)$,
where $\Sigma_d$ is the signature of data symbols, each of type $d^m \to d$
(then the data arity of such a symbol $g$ is defined to be $\ar_d(g) = m$),
$\Sigma_s$ is the signature of structure symbols~$f$,
which have types of the shape $d^m \times s^n \to s$
(and data arity $\ar_d(f) = m$, structure arity $\ar_s(f) = n$),
$\C \subseteq \Sigma_s$ is a set of \emph{constructors},
$\R_d$ is a terminating TRS over the signature $\Sigma_d$, and
$\R_s$ is a TRS over the signature $\Sigma_d \cup \Sigma_s$,
containing rules $f(u_1, \dotsc, u_m, t_1, \dotsc, t_n) \to t$
that satisfy the following properties:
\begin{itemize}
\item
$f \in \Sigma_s \setminus \C$ with $\ar_d(f) = m$, $\ar_s(f) = n$,

\item
$f(u_1, \dotsc, u_m, t_1, \dotsc, t_n)$ is a well-sorted linear term,

\item
$t$ is a well-sorted term of sort~$s$, and

\item
for all $1 \le i \le n$ and for all $p \in \Pos(t_i)$ such that $t_i|_p$ is not
a variable and $\rt(t_i|_p) \in \Sigma_s$, it holds that
$\rt(t_i|_{p'}) \notin \C$ for all $p' < p$
(i.e., no structure symbol is below a constructor).
\end{itemize}

Furthermore, $\R_s$ is required to be \emph{exhaustive}, meaning that for every
$f \in \Sigma_s \setminus \C$ with $\ar_d(f) = m$, $\ar_s(f) = n$, ground normal
forms $u_1, \dotsc, u_m \in \NF_{\gnd}(\R_d)$, and terms
$t_1, \dotsc, t_n \in \T(\Sigma_d \cup \Sigma_s)$ such that for every
$1 \le i \le n$,
$t_i = c_i(u_1', \dotsc, u_k', t_1', \dotsc, t_l')$ with
$u_j' \in \NF_\gnd(\R_d)$ for $1 \le j \le k = \ar_d(c_i)$ and
$c_i \in \C$, there exists at least one rule $\ell \to r \in \R_s$ such that
$\ell$ matches the term $f(u_1, \dotsc, u_m, t_1, \dotsc, t_n)$.

A proper specification $\Spec$ is called \emph{orthogonal}, if $\R_d \cup \R_s$
is orthogonal, otherwise it is called \emph{non-orthogonal}.
\end{definition}

The above definition coincides with the definition of proper specifications
given in~\cite{ZR10} for orthogonal proper specifications.\footnote{%
To see this, one should observe that a defined symbol cannot occur on a non-root
position of a left-hand side. This holds since otherwise the innermost such
symbol would have variables and constructors as structure arguments and data
arguments that do not unify with any of the data rules (due to orthogonality),
which therefore are normal forms and can be instantiated to ground normal forms.
Thus, exhaustiveness would require a left-hand side to match this term when
instantiating all structure variables with some terms having a constructor root,
which would give a contradiction to non-overlappingness.%
}
We will illustrate the restrictions in the above definition later in
Section~\ref{sec:Criteria}.
In the following, all examples except for Example~\ref{ex:CexpLeftLin}
will be using the domain of Boolean streams,
where $\C = \{ {:} \}$ and $\Sigma_d \supseteq \{ \zero,\one \}$ with
$\ar_d(\zero) = \ar_d(\one) = 0$ and $\ar_d(:) = \ar_s(:) = 1$.
In these examples, only a data TRS~$\R_d$ and a structure TRS~$\R_s$ are given
from which the remaining symbols in~$\Sigma_d$ and~$\Sigma_s$ and their arities
can be derived.
If the data TRS~$\R_d$ is not provided it is assumed to be empty.

\section{Productivity}
\label{sec:Productivity}

For orthogonal proper specifications, productivity is the property that
every ground term $t$ of sort~$s$ can, in the limit, be rewritten to a possibly
infinite term consisting only of constructors. This is equivalent to stating
that for every prefix depth $k \in \Nats$, the term $t$ can be rewritten to
another term $t'$ having only constructor symbols on positions of depth $k$ or
less.

\begin{definition}
\label{def:OrthoProd}
An orthogonal proper specification $\Spec = (\Sigma_d,\Sigma_s,\C,\R_d,\R_s)$ is
\emph{productive},
iff for every ground term $t$ of sort~$s$ and every $k \in \Nats$,
there is a reduction $t \to_{\R_d \cup \R_s}^* t'$ such that every
symbol of sort~$s$ in $t'$ on depth less or equal to~$k$ is a constructor.
\end{definition}

Productivity of an orthogonal proper specification is equivalent to the
following property, as was shown in~\cite{ZR10}.

\begin{proposition}
\label{prop:OrthoProd}
An orthogonal proper specification $\Spec = (\Sigma_d,\Sigma_s,\C,\R_d,\R_s)$ is 
\emph{productive},
iff for every ground term $t$ of sort~$s$ there is a reduction
$t \to_{\R_d \cup \R_s}^* t'$ such that $\rt(t') \in \C$.
\end{proposition}

It was already observed in~\cite{EGH09,End10} that productivity of orthogonal
specifications is equivalent to the existence of an \emph{outermost-fair}
reduction computing a constructor prefix for any given depth. Below, we give a
general definition of outermost-fair reductions, as they will
also be used in the non-orthogonal setting.

\pagebreak[2]
\begin{definition}
\label{def:OutFair}
~
\begin{itemize}
\item
A \emph{redex} is a subterm $t|_p$ of a term $t$ at position $p \in \Pos(t)$
such that a rule $\ell \to r$ and a substitution $\sigma$ exist with
$t|_p = \ell\sigma$.
The redex $t|_p$ is said to be \emph{matched} by the rule $\ell \to r$.

\item
A redex is called \emph{outermost} iff it is not a strict subterm of another
redex.

\item
A redex $t|_p = \ell\sigma$ is said to \emph{survive} a reduction step
$t \to_{\ell' \to r', q} t'$ if $p \parallel q$, or if $p < q$ and
$t' = t[\ell\sigma']_p$ for some substitution $\sigma'$
(i.e., the same rule can still be applied at $p$).

\item
A rewrite sequence (reduction) is called \emph{outermost-fair}, iff there is no
outermost redex that survives as an outermost redex infinitely long.

\item
A rewrite sequence (reduction) is called \emph{maximal}, iff it is infinite or
ends in a \emph{normal form} (a term that cannot be rewritten further).
\end{itemize}
\end{definition}

For non-orthogonal proper specifications, requiring just the existence of a
reduction to a normal form (or to a constructor prefix of arbitrary depth)
does not guarantee the computation to reach it, due to the possible
non-deterministic choices. This can be observed for the term $\fsym{maybe}$
in the following example.

\begin{example}
\label{ex:NonOrthoMaybeRnd}
Consider a proper specification with the TRS $\R_s$ consisting of the following
rules:
\[
\begin{array}{rcl@{\qquad\qquad}rcl}
    \fsym{maybe} &\to& \zero : \fsym{maybe}
&
    \fsym{random} &\to& \zero : \fsym{random}
\\
    \fsym{maybe} &\to& \phantom{\zero:\fsym{maybe}}\mathllap{\fsym{maybe}}
&
    \fsym{random} &\to& \one : \fsym{random}
\end{array}
\]

This specification is not orthogonal, since the rules for $\fsym{maybe}$ as well
as those for $\fsym{random}$ overlap.
We do not want to call this specification productive, since it admits the
infinite outermost-fair reduction $\fsym{maybe} \to \fsym{maybe} \to \ldots$
that never produces any constructors. However, there exists an infinite
reduction producing infinitely many constructors starting in the term
$\fsym{maybe}$, namely
$\fsym{maybe} \to \zero : \fsym{maybe}
\to \zero : \zero : \fsym{maybe} \to \ldots$.
When only considering the rules for $\fsym{random}$ then we want to call
the resulting specification productive, since no matter what rule of
$\fsym{random}$ we choose, an element of the stream is created.
\end{example}

Requiring just the existence of a constructor normal form is called
\emph{weak productivity} in~\cite{EGH09,End10}. We already stated above that
this is not the notion of productivity we are interested in. The one we are
interested in is \emph{strong productivity}, which is also defined
in~\cite{EGH09,End10}, since it requires all reductions that make progress on
outermost positions to reach constructor normal forms.

\begin{definition}
\label{def:StrongProd}
A proper specification $\Spec$ is called \emph{strongly productive} iff for
every ground term $t$ of sort~$s$ all maximal outermost-fair rewrite sequences
starting in $t$ end in (i.e., have as limit for infinite sequences)
a constructor normal form.
\end{definition}

It was observed in~\cite{EGH09,End10} that weak and strong productivity coincide
for orthogonal (proper) specifications.
However, for non-orthogonal (proper) specifications this is not the case
anymore.
The rules for $\fsym{maybe}$ in Example~\ref{ex:NonOrthoMaybeRnd} are not
strongly productive, since they allow the infinite outermost-fair reduction
$\fsym{maybe} \to \fsym{maybe} \to \ldots$. However, these rules are weakly
productive, since any ground term can be rewritten to an infinite stream
containing only $\zero$~elements after some finite prefix.
For example, the ground term $\one : \fsym{maybe}$ can be rewritten to the
infinite stream $\one : \zero : \zero : \ldots$.

An example of a non-orthogonal proper specification that is both strongly and
weakly productive are the rules for $\fsym{random}$ in
Example~\ref{ex:NonOrthoMaybeRnd}, which always produce an infinite stream.
In this case, the restriction to outermost-fair reductions is not needed.
However, if we add the rule $\fsym{id}(xs) \to xs$ and replace the
rule $\fsym{random} \to \one : \fsym{random}$ by the rule
$\fsym{random} \to \fsym{id}(\one : \fsym{random})$, then
the infinite reduction
$\fsym{random} \to \fsym{id}(\one : \fsym{random}) \to
\fsym{id}(\one : \fsym{id}(\one : \fsym{random})) \to \ldots$ exists. This
reduction is not outermost-fair since the outermost redex $\fsym{id}(\ldots)$
survives infinitely often.
When restricting to outermost-fair reductions, then indeed an infinite stream of
Boolean values is obtained for every such reduction, so this is a strongly
productive proper specification, too.
Note that strong productivity implies weak productivity, so the example is also
weakly productive.

\section{Criteria for Strong Productivity}
\label{sec:Criteria}

For orthogonal proper specifications, it is sufficient to just consider
reductions that create a constructor at the top, as stated in
Proposition~\ref{prop:OrthoProd}. We will show next that this is also the case
for non-orthogonal proper specifications. However, in contrast to~\cite{ZR10},
here we have to consider all maximal outermost-fair reductions, instead of just
requiring the existence of such a reduction.

\begin{proposition}
\label{prop:NonOrthoProd}
A proper specification $\Spec = (\Sigma_d,\Sigma_s,\C,\R_d,\R_s)$ is
\emph{strongly productive} iff for every maximal outermost-fair reduction
$t_0 \to_{\R_d \cup \R_s} t_1 \to_{\R_d \cup \R_s} \dots$
with $t_0$ being of sort~$s$ there exists $k \in \Nats$ such that
$\rt(t_k) \in \C$.
\end{proposition}

\begin{proof}
The ``only if''-direction is trivial. For the ``if''-direction, we show
inductively that for every depth $z \in \Nats$ and every maximal outermost-fair
reduction $\rho \equiv t_0 \to_{p_0} t_1 \to_{p_1} \dots$ there exists an index
$j \in \Nats$ such that for all positions $p \in \Pos(t_j)$ of sort~$s$ with
$|p| < z$, $\rt(t_j|_p) \in \C$.

For $z = 0$, the index $j$ can be set to $0$, thus here the claim trivially
holds.
Otherwise, we get that an index $k \in \Nats$ exists such that
$\rt(t_k) \in \C$. Let $t_k = c(u_1', \dotsc, u_m', t_1', \dotsc, t_n')$ with
$c \in \C$.
Because $c$ is a constructor, we know that $p_l > \epsilon$ for all $l \ge k$.
Define $P_r = \{ p_i' \mid p_i = (m{+}r) . p_i' \}$ for $1 \le r \le n$ (i.e.,
the positions in the maximal outermost-fair reduction that are occurring in
structure argument~$r$).
Then, for $1 \le r \le n$ and $P_r = \{ p_0^r, p_1^r, \dotsc \}$ the
reduction $t_r' = t_{r,0} \to_{p_0^r} t_{r,1} \to_{p_1^r} \dots$ is also a
maximal outermost-fair reduction, otherwise an infinitely long surviving
outermost redex would also be an infinitely long surviving outermost redex of
the reduction~$\rho$. By the induction hypothesis for $z-1$ we get
that indices $j_r$ for $1 \le r \le n$ exist such that
$\rt(t_{r,j_r}|_p) \in \C$ for all positions $p \in \Pos(t_{r,j_r})$ with
$|p| < z-1$. Since all these reductions were taken from the original reduction,
we define $j = k + \#\text{d-red} + \sum_{i=1}^n j_i$,
where $\#\text{d-red}$ denotes the number of reductions performed in the
data arguments of the constructor~$c$ such that $p_j = p_{j_r}^r$ for the
last~$r$.
This shows that the initial reduction~$\rho$ has the form
$t_0 \to^* t_k = c(u_1', \dotsc, u_m', t_1', \dotsc, t_n') \to^*
c(u_1'', \dotsc, u_m'', t_1'', \dotsc, t_n'') = t_{j+1}$,
where $t_{r,j_r} \to^* t_r''$ for every $1 \le r \le n$. Since there are only
constructors in $t_{r,j_r}$ for depths $0, \dotsc, z-2$, these constructors are
still present in $t_r''$. This proves the proposition, since $c \in \C$ and
thus for all positions $p \in \Pos(t_j)$ of sort~$s$ with $|p| < z$ we have
$\rt(t_j|_p) \in \C$.
\end{proof}

This characterization of strong productivity will be used in the remainder of
the paper. Note that it is similar to the requirements for infinitary strong
normalization $\SNinf$ observed in~\cite{Z08}, where it is found that for
left-linear and finite term rewrite systems, $\SNinf$ holds if and only if every
infinite reduction only contains a finite number of root steps. Thus, it could
seem possible to define strong productivity of proper specifications by
requiring that every reduction starting in a finite ground term is infinitary
strongly normalizing, i.e., $\SNinf$ holds for the relation
${\to_{\R_d \cup \R_s}} \cap \T(\Sigma_d \cup \Sigma_s)^2$.
However, this is not the case, as the following example shows.

\begin{example}
\label{ex:CexpSNinf}
Consider the proper specification containing the following TRS $\R_s$:
\[
\begin{array}{rcl@{\qquad\qquad}rcl}
\fsym{a} &\to& \fsym{f}(\fsym{a})
&
\fsym{f}(x : xs) &\to& x : xs
\end{array}
\]

This TRS has the property $\SNinf$, intuitively because either the
symbol~$\fsym{f}$ remains at the root position and can never be rewritten again
(in case the first rule is applied), or the constructor~$:$ is created at the
root. Formally, this can for example be proven by the technique presented
in~\cite{Z08}: Let $\Sigma_{\#} = \Sigma \uplus \{ g_{\#} \mid g \in \Sigma\}$,
where $\Sigma = \{ \zero, \one, {:}, \fsym{a}, \fsym{f} \}$ is the signature of
the specification.
Then we choose the finite weakly monotone $\Sigma_{\#}$ algebra
$(\{0,1,2\}, [\cdot], \bot, >, \ge)$, where $\bot = 0$, $[\zero] = 0$,
$[\one] = 0$,
$[\fsym{a}] = 1$, $[\fsym{f}](n) = n$, $[:](m,n) = \min \{m+n, 2\}$,
$[a_{\#}] = 2$, $[f_{\#}](n) = 1$, and $[:_{\#}](m,n) = 0$ for
$m,n \in \{0,1,2\}$ and~$>$ and~$\ge$ are the natural comparison operators on
the numbers $\{0,1,2\}$.
It is easy to check that this algebra is indeed weakly monotone (i.e., that
$>$~is well-founded, ${>} \cdot {\ge} \subseteq {>} \subseteq {\ge}$, and
for every $g \in \Sigma_{\#}$, the operation $[g]$ is monotone with respect
to~$\ge$). Additionally, the requirements of the combination
of~\cite[Theorem~5 and Theorem~6]{Z08} are satisfied, i.e., $\{0,1,2\}$ is
finite, $\ge$~is transitive, $a \ge b$ implies $a > b$ or $a = b$,
$a \ge \bot = 0$ for all $a,b \in \{0,1,2\}$, and
$[\ell\,\sigma] \ge [r\,\sigma]$ and $[\ell_{\#}\,\sigma] > [r_{\#}\,\sigma]$
for all $\ell \to r \in \R_s$ and all substitutions $\sigma$, where
$g(t_1, \dotsc, t_k)_{\#} = g_{\#}(t_1, \dotsc, t_k)$.
This proves $\SNinf$ of $\to_{\R_s}$, which especially entails $\SNinf$ of the
relation ${\to_{\R_s}} \cap \T(\Sigma_d \cup \Sigma_s)^2$.

However, the above proper specification is not strongly productive, since the
infinite outermost-fair reduction
$\fsym{a} \to_{\R_s} \fsym{f}(\fsym{a}) \to_{\R_s} \fsym{f}(\fsym{f}(\fsym{a}))
\to_{\R_s} \ldots$,
continued by repeatedly reducing the symbol $\fsym{a}$,
never produces any constructors.
\end{example}

The above example shows that even though we require exhaustiveness of proper
specifications, this exhaustiveness only refers to constructor terms, i.e., the
objects we are interested in, and not to arbitrary terms. A similar observation,
namely that top termination is not equivalent to productivity, was already made
in~\cite{ZR09}.

A first technique to prove strong productivity of proper specifications is given
next. It is a simple syntactic check that determines whether every right-hand
side of sort~$s$ starts with a constructor. For orthogonal proper
specifications, this was already observed in~\cite{ZR10}. It has to be proven
again since here we consider strong productivity, which requires all possible
outermost-fair reductions to reach a constructor normal form, instead of weak
productivity as in~\cite{ZR10}, for which only a single reduction to a
constructor normal form needs to be constructed.

\begin{theorem}
\label{thm:SyntacticCheck}
Let $\Spec = (\Sigma_d,\Sigma_s,\C,\R_d,\R_s)$ be a proper specification.
If for all rules $\ell \to r \in \R_s$ we have $\rt(r) \in \C$, then $\Spec$ is
strongly productive.
\end{theorem}

\begin{proof}
Let $\rho \equiv t_0 \to_{p_0} t_1 \to_{p_1} \dots$ be a maximal outermost-fair
reduction and let $t_0 = f(u_1', \dotsc, u_m', t_1', \dotsc, t_n')$.
If $f \in \C$ we are done, so we assume $f \in \Sigma_s \setminus \C$ and
perform structural induction on $t_0$ to prove that $\rt(t_k) \in \C$ for some
$k \in \Nats$.

From the induction hypothesis we get that for every $1 \le i \le n$ and every
maximal outermost-fair reduction $t_i' = t_{i,0} \to t_{i,1} \to \dots$ there
exists an index $k_i \in \Nats$ such that $\rt(t_{i,k_i}) \in \C$.

Assume that for all $j \in \Nats$, $p_j \ne \epsilon$.
As in the proof of Proposition~\ref{prop:NonOrthoProd}, we therefore again
obtain maximal outermost-fair reductions $t_i' \to \dots$, thus we get indices
$k_i \in \Nats$ such that $\rt(t_{i,k_i}) \in \C$, as explained above. This
makes our reduction $\rho$ have the shape
$t_0 = f(u_1', \dotsc, u_m', t_1', \dotsc, t_n') \to^*
f(u_1'', \dotsc, u_m'', t_1'', \dotsc, t_n'') = t_j$ for some $j \in \Nats$,
where $u_1'', \dotsc, u_m'' \in \NF_\gnd(\R_d)$ (since the reduction $\rho$ is
maximal outermost-fair and $\R_d$ is terminating) and
$t_{i,k_i} \to^* t_i''$, thus also $\rt(t_i'') \in \C$. Because $\R_s$ is
exhaustive, we get that $t_j$ contains a redex at the root position~$\epsilon$,
which of course is outermost. This gives rise to a contradiction to $\rho$ being
outermost fair, as this outermost redex survives infinitely often, because
$p_j \ne \epsilon$ for all $j \in \Nats$. Therefore, $p_j = \epsilon$ for some
$j \in \Nats$ and the reduction has the shape $t_0 \to^* t_j \to_{\epsilon}
r\sigma$, where the last step is with respect to some rule
$\ell \to r \in \R_s$.
By the assumption on the shape of the rules in $\R_s$, we have $\rt(r) \in \C$,
hence also $\rt(r\sigma) \in \C$, which proves productivity according to
Proposition~\ref{prop:NonOrthoProd}.
\end{proof}

This technique is sufficient to prove strong productivity of the proper
specification consisting of the two rules for $\fsym{random}$ in
Example~\ref{ex:NonOrthoMaybeRnd}, since both have right-hand sides with the
constructor~$:$ at the root. However, it is easy to create examples which are
strongly productive, but do not satisfy the syntactic requirements of
Theorem~\ref{thm:SyntacticCheck}.

\begin{example}
\label{ex:NonSyntax}
Consider the proper specification with the following TRS $\R_s$:
\[
\begin{array}{rcl@{\qquad}rcl}
    \fsym{ones} &\to& \one : \fsym{ones}
&
    \fsym{finZeroes} &\to& \zero : \fsym{ones}
\\
    \fsym{finZeroes} &\to& \zero : \zero : \fsym{ones}
&
    \fsym{finZeroes} &\to& \zero : \zero : \zero : \fsym{ones}
\\
    \fsym{f}(\zero:xs) &\to& \fsym{f}(xs)
&
    \fsym{f}(\one:xs) &\to& \one : \fsym{f}(xs)
\end{array}
\]

The constant $\fsym{finZeroes}$ produces non-deterministically a stream that
starts with one, two, or three zeroes followed by an infinite stream of ones.
Function~$\fsym{f}$ takes a binary stream as argument and filters out all
occurrences of zeroes. Thus, productivity of this example proves that only a
finite number of zeroes can be produced. This however cannot be proven with the
technique of Theorem~\ref{thm:SyntacticCheck}, since the right-hand side of the
rule $\fsym{f}(\zero:xs) \to \fsym{f}(xs)$ does not start with the
constructor~$:$.
\end{example}

Another technique presented in~\cite{ZR10} to show productivity of orthogonal
proper specifications is based on context-sensitive termination~\cite{L98}.
The idea is to disallow rewriting in structure arguments of constructors, thus
context-sensitive termination implies that for every ground term of sort~$s$,
a term starting with a constructor can be reached (due to the exhaustiveness
requirement). As was observed by Endrullis and Hendriks recently in~\cite{EH11},
this set of blocked positions can be enlarged, making the approach even
stronger.

Below,
the technique for proving productivity by showing termination of a corresponding
context-sensitive TRS is extended
to also be applicable in the case of our more general proper specifications.
This version already includes an adaption of the improvement mentioned above.

\begin{definition}
\label{def:Mu}
Let $\Spec = (\Sigma_d,\Sigma_s,\C,\R_d,\R_s)$ be a proper specification.
The replacement map $\mu_\Spec : \Sigma_d \cup \Sigma_s \to 2^{\Nats}$
is defined as follows:
\footnote{
Note that in~\cite{EH11}, Endrullis and Hendriks consider orthogonal TRSs and
also block arguments of symbols in $\Sigma_d$ which only contain variables.
This however is problematic when allowing data rules that are not left-linear.
Example:
\[
\begin{array}{l@{\qquad}rcl@{\qquad\qquad}rcl@{\qquad\qquad}rcl}
\R_s:&
\fsym{f}(\fsym{1}) &\to&\fsym{f}(\fsym{d}(\fsym{0},\fsym{d}(\fsym{1},\fsym{0})))
&
\fsym{f}(\fsym{0}) &\to&  \fsym{0} : \fsym{f}(\fsym{0})
\\
\R_d:&
\fsym{d}(x,x) &\to& \fsym{1}
&
\fsym{d}(\fsym{0},x) &\to& \fsym{0}
&
\fsym{d}(\fsym{1},x) &\to& \fsym{0}
\end{array}
\]
Here, the term $\fsym{f}(\fsym{d}(\fsym{0},\fsym{d}(\fsym{1},\fsym{0})))$ can
only be $\mu$-rewritten to the term $\fsym{f}(\fsym{0})$ (which then in turn has
to be rewritten to $\fsym{0} : \fsym{f}(\fsym{0})$)
if defining $\mu(\fsym{d}) = \{ 1 \}$, since the subterm
$\fsym{d}(\fsym{1},\fsym{0})$ can never be rewritten to $\fsym{0}$. However, the
example is not strongly productive, as reducing in this way gives rise to an
infinite outermost-fair reduction
$
\fsym{f}(\fsym{d}(\fsym{0},\fsym{d}(\fsym{1},\fsym{0})))
\to
\fsym{f}(\fsym{d}(\fsym{0},\fsym{0}))
\to
\fsym{f}(\fsym{1})
\to
\dots
$.
Blocking arguments of data symbols can only be done when $\R_d$ is
left-linear.
}
\begin{itemize}
\item
$\mu_\Spec(f) = \{ 1, \dotsc, \ar_d(f) \}$, if $f \in \Sigma_d \cup \C$

\item
$\mu_\Spec(f) = \{ 1, \dotsc, \ar_d(f) + \ar_s(f) \}
    \setminus
    \{
        1 \le i \le \ar_d(f) + \ar_s(f)
        \mid
        t|_i$ is a variable
        for all $\ell \to r \in \R_s$
        and all non-variable subterms $t$ of $\ell$ with $\rt(t) = f
    \}$,\footnote{
        The requirement of $t$ not being a variable ensures that
        $\rt(t)$ is defined.
    }
    otherwise
\end{itemize}
\end{definition}

In the remainder, we leave out the subscript $\Spec$ if the specification is
clear from the context. The replacement map $\mu$ is used to define
the set of \emph{allowed} positions of a non-variable term $t$ as
$\Pos_\mu(t) = \{ \epsilon \} \cup
    \{ i.p \mid i \in \mu(\rt(t)), p \in \Pos_\mu(t|_i) \}$
and the set of \emph{blocked} positions of $t$ as
$\blocked(t) = \Pos(t) \setminus \Pos_\mu(t)$.
Context-sensitive rewriting~\cite{L98} then
is the restriction of the rewrite relation to those redexes on positions
from $\Pos_\mu$. Formally, we have $t \muto_{\ell \to r,p} t'$ iff
$t \to_{\ell \to r,p} t'$ and $p \in \Pos_\mu(t)$ and we say a TRS $\R$ is
\emph{$\mu$-terminating} iff no infinite $\muto_{\R}$-chain exists.

The replacement map $\mu_{\Spec}$ is \emph{canonical}~\cite{L02} for the
left-linear TRS $\R_s$, guaranteeing through the second condition of the above
Definition~\ref{def:Mu} that non-variable positions of left-hand
sides are allowed.
In that definition, the replacement map $\mu_{\Spec}$ is extended to the
possibly non-left-linear TRS $\R_d \cup \R_s$ by allowing all arguments of
symbols from~$\Sigma_d$.

Our main result of this paper is that also for possibly non-orthogonal proper
specifications, $\mu$-termination implies productivity.

\begin{theorem}
\label{thm:CStermination}
A proper specification $\Spec = (\Sigma_d,\: \Sigma_s,\: \C,\: \R_d,\: \R_s)$ is
strongly productive, if $\R_d \cup \R_s$ is $\mu_{\Spec}$-terminating.
\end{theorem}

Before proving the above theorem, we will show first that it subsumes
Theorem~\ref{thm:SyntacticCheck}. Intuitively, this holds because structure
arguments of constructors are blocked, and if every right-hand side of $\R_s$
starts with a constructor then the number of allowed redexes of sort~$s$ in a
term steadily decreases.

\begin{proposition}
\label{prop:CSimpliesSyn}
Let $\Spec = (\Sigma_d, \Sigma_s, \C, \R_d, \R_s)$ be a proper specification.
If for all rules $\ell \to r \in \R_s$ we have $\rt(r) \in \C$, then
$\R_d \cup \R_s$ is $\mu_{\Spec}$-terminating.
\end{proposition}

\begin{proof}
Let $t \in \T(\Sigma_d \cup \Sigma_s, \V)$ be well-typed.
If $t$ has sort~$d$, then all subterms must also be of sort~$d$, as symbols from
$\Sigma_d$ only have arguments of that sort. Hence, rewriting can only be done
with rules from $\R_d$, which is assumed to be terminating.

Otherwise, let $t$ be of sort~$s$ and assume that~$t$ starts an infinite
$\mu$-reduction
$t = t_0 \muto_{\ell_0 \to r_0, p_0} t_1 \muto_{\ell_1 \to r_1, p_1} t_2
    \muto_{\ell_2 \to r_2, p_2} \dots$.
We define
$\PosRedS_\mu(t') = \{ p \in \Pos_\mu(t') \mid t'|_p$ is a redex of sort~$s \}$
for any term $t' \in \T(\Sigma_d \cup \Sigma_s, \V)$.
It will be proven that in every step $t_i \muto_{\ell_i \to r_i, p_i} t_{i+1}$
of the infinite reduction,
$\lvert\PosRedS_\mu(t_{i+1})\rvert \le \lvert\PosRedS_\mu(t_i)\rvert$
and that for steps with $\ell_i \to r_i \in \R_s$, we even have
$\lvert\PosRedS_\mu(t_{i+1})\rvert < \lvert\PosRedS_\mu(t_i)\rvert$.
To this end, case analysis of the rule $\ell_i \to r_i$ is performed.
If $\ell_i \to r_i \in \R_d$, then $t_i = t_i[\ell_i\sigma_i]_{p_i}$ and
$t_{i+1} = t_i[r_i\sigma_i]_{p_i}$ for some substitution $\sigma_i$.
Because $\ell_i, r_i \in \T(\Sigma_d,\V)$,
$\lvert\PosRedS_\mu(\ell_i\sigma_i)\rvert
= \lvert\PosRedS_\mu(r_i\sigma_i)\rvert = 0$
since all symbols in $\Sigma_d$ have arguments of sort~$d$.
Thus, $\PosRedS_\mu(t_{i+1}) = \PosRedS_\mu(t_i)$.
In the second case, $\ell_i \to r_i \in \R_s$. Let
$t_i = t_i[\ell_i\sigma_i]_{p_i}$ and
$t_{i+1} = t_i[r_i\sigma_i]_{p_i}$ for some substitution $\sigma_i$.
Then, $\PosRedS_\mu(t_i) =
    \PosRedS_\mu(t_i[z]_{p_i})
    \uplus \{ p_i.p \mid p \in \PosRedS_\mu(t_i|_{p_i}) \}$
for any variable $z \in \V$ of sort~$s$.
For $t_{i+1}$ we observe that
$\PosRedS_\mu(t_{i+1}) = \PosRedS_\mu(t_i[r_i\sigma_i]_{p_i}) =
    \PosRedS_\mu(t_i[z]_{p_i})
    \uplus \{ p_i.p \mid p \in \PosRedS_\mu(t_i[r_i\sigma_i]_{p_i}|_{p_i}) \}$
for any variable $z \in \V$ of sort~$s$.
Here, it holds that
$\PosRedS_\mu(t_i|_{p_i}) = \PosRedS_\mu(\ell_i\sigma_i) \ni \epsilon$,
therefore $p_i \in \PosRedS_\mu(t_i)$.
Furthermore, $\PosRedS_\mu(t_i[r_i\sigma_i]_{p_i}|_{p_i}) =
    \PosRedS_\mu(r_i\sigma_i) = \emptyset$,
since $\rt(r_i) \in \C$ by assumption, hence
$\mu(\rt(r_i)) = \{ 1, \dotsc, \ar_d(\rt(r_i)) \}$ and because symbols from
$\Sigma_d$ only have arguments of sort~$d$.
Thus, $\PosRedS_\mu(t_{i+1}) \subsetneq \PosRedS_\mu(t_i)$.

Combining these observations, we therefore only have finitely many reductions
with rules from $\R_s$ in the infinite reduction. Thus, an infinite tail of
steps with rules from $\R_d$ exists. This however contradicts the assumption
that $\R_d$ is terminating, hence no infinite $\mu$-reduction can exist which
proves $\mu$-termination of $\R_d \cup \R_s$.
\end{proof}

Hence, we could restrict ourselves to analyzing context-sensitive termination
only. However, the syntactic check of Theorem~\ref{thm:SyntacticCheck} can be
done very fast and should therefore be the first method to try.

In order to prove Theorem~\ref{thm:CStermination} we will show that a maximal
outermost-fair reduction that never reaches a constructor entails an infinite
$\mu$-reduction.
For this purpose we need the following lemma,
which shows that in every ground term not starting with a constructor there
exists a redex that is not blocked by the replacement map~$\mu$.

\begin{lemma}
\label{lem:ExistRedex}
Let $\Spec = (\Sigma_d,\Sigma_s,\C,\R_d,\R_s)$ be a proper specification.
For all ground terms $t$ of sort $s$ with $\rt(t) \notin \C$ there exists a
position $p \in \Pos_\mu(t)$ such that $t \to_{p}$.
\end{lemma}

\begin{proof}
Let $t = f(u_1, \dotsc, u_m, t_1, \dotsc, t_n)$. We perform structural induction
on $t$.
If $u_i \to_{p'}$ for some $1 \le i \le m$ with $i \in \mu(f)$, then
$t \to_{i.p'}$ and $i.p' \in \Pos_\mu(t)$ since arguments of data symbols are
never blocked.
Thus, we assume in the remainder that
$u_i \in \NF_\gnd(\R_d)$ for all $1 \le i \le m$ with $i \in \mu(f)$.
If $\rt(t_i) \in \C$ for all $1 \le i \le n$, $i \in \mu(f)$, then
$t \to_{\epsilon}$ by the exhaustiveness requirement (and because all arguments
$u_j$, $t_j$ with $j \notin \mu(f)$ are being matched by pairwise different
variables, due to left-linearity).
Otherwise, there exists $1 \le i \le n$, $i \in \mu(f)$ such that
$\rt(t_i) \notin \C$. By the induction hypothesis we get that $t_i \to_{p'}$
for some $p' \in \Pos_\mu(t_i)$. Therefore, we also have
$i.p' \in \Pos_\mu(t)$ and $t \to_{i.p'}$.
\end{proof}

A second lemma that is required for the proof of Theorem~\ref{thm:CStermination}
states that a specialized version of the
Parallel Moves Lemma~\cite[Lemma~6.4.4]{BN98}
holds for our restricted format of term rewrite systems.
It allows us to swap the order of reductions blocked by~$\mu$ with reductions
not blocked by~$\mu$.
To formulate the lemma, we need the notion of a parallel reduction step
$t \parto_P t'$, which is defined for a set
$P = \{ p_1, \dotsc, p_n \} \subseteq \Pos(t)$ such that for every pair
$1 \le i < j \le n$ we have $p_i \parallel p_j$
and a term
$t = t[\ell_1\sigma_1]_{p_1} \ldots [\ell_n\sigma_n]_{p_n}$ as
$t' = t[r_1\sigma_1]_{p_1} \ldots [r_n\sigma_n]_{p_n}$ for rules
$\ell_i \to r_i \in \R_d \cup \R_s$ and substitutions $\sigma_i$,
$1 \le i \le n$.

\begin{lemma}
\label{lem:SpecialPML}
Let $\Spec = (\Sigma_d,\Sigma_s,\C,\R_d,\R_s)$ be a proper specification.
For all ground terms $t, t', t''$ and positions $p \in \Pos_\mu(t')$,
$P \subseteq \blocked(t)$ with $t \parto_P t' \to_{\ell \to r, p} t''$,
a term $\hat{t}$ and a set $P' \subseteq \Pos(\hat{t})$ exist such that
$t \to_{\ell \to r, p} \hat{t} \parto_{P'} t''$.
\end{lemma}

\begin{proof}
Let $P = \{ p_1, \dotsc, p_k \} \subseteq \blocked(t)$.
Then $t = t [\ell_1\sigma_1]_{p_1} \ldots [\ell_k\sigma_k]_{p_k}
\parto_P
t [r_1\sigma_1]_{p_1} \ldots [r_k\sigma_k]_{p_k} = t'
= t'[\ell\sigma]_{p}
$ for some rules
$\ell_1 \to r_1, \dotsc, \ell_k \to r_k, \ell \to r \in \R_d \cup \R_s$ and
substitutions $\sigma_1, \dotsc, \sigma_k, \sigma$.
W.l.o.g., let $0 \le j \le k$ be such that
$p_i \not\parallel p$ for all $1 \le i \le j$
and $p_i \parallel p$ for all $j < i \le k$.
Since $p \in \Pos_\mu(t')$ and $p_i \in \blocked(t')$, it must
hold that $p < p_i$ for all $1 \le i \le j$.
Therefore, the term $t'$ must have the shape
$
t' = t\left[ \ell\sigma
        [r_1\sigma_1]_{p_1 - p} \ldots [r_j\sigma_j]_{p_j-p}
\right]_{p}\;
    [r_{j+1}\sigma_{j+1}]_{p_{j+1}} \ldots [r_k\sigma_k]_{p_k}
$.

If $\ell \to r \in \R_d$, then it must hold that $j=0$, since arguments of data
symbols are never blocked. Hence, the lemma trivially holds in this case, as all
reductions are on independent positions.

Otherwise, $\ell \to r \in \R_s$.
Because the positions $p_i$ for $1 \le i \le j$ are blocked, it must be the case
that they are either below a variable in all rules containing a certain
symbol~$f$ (hence, they are also below a variable in~$\ell$), or they are
below a structure argument of a constructor $c \in \C$. By requirement of
specifications, if a constructor is present on a left-hand side of a rule, all
its structure arguments must be variables. Thus, we conclude that all positions
$p_i$, and thereby all terms $r_i\sigma_i$, are below some variable of
$\ell$ in $t'$. Additionally, the left-hand side $\ell$ is required to be
linear, therefore there exist pairwise different variables $x_1, \dotsc, x_j$,
contexts $C_1, \dotsc, C_j$, and a substitution $\sigma'$ being like $\sigma$
except that $\sigma'(x_i) = x_i$ for $1 \le i \le j$ such that:
\[
\begin{array}{r@{\;\;}l@{\;\;}l@{\;\;}c@{\;\;}l}
t' &=&
t\left[ \ell\sigma'
        \{
            x_1 {:=} C_1[r_1\sigma_1], \dotsc, x_j {:=} C_j[r_j\sigma_j]
        \}
\right]_{p}\;
    [r_{j+1}\sigma_{j+1}]_{p_{j+1}} \ldots [r_k\sigma_k]_{p_k}
\\
&\to_{p}&
t\left[ r\sigma'
        \{
            x_1 {:=} C_1[r_1\sigma_1], \dotsc, x_j {:=} C_j[r_j\sigma_j]
        \}
\right]_{p}\;
    [r_{j+1}\sigma_{j+1}]_{p_{j+1}} \ldots [r_k\sigma_k]_{p_k}
&=& t''
\end{array}
\]

We conclude that $p \in \Pos_\mu(t)$, as all reduction steps in
$t \parto_P t'$ are either below or independent of $p$. Thus:
\[
\begin{array}{r@{\;\;}l@{\;\;}l@{\;\;}c@{\;\;}l}
t &\,=&
t\left[ \ell\sigma'
        \{
            x_1 {:=} C_1[\ell_1\sigma_1], \dotsc, x_j {:=} C_j[\ell_j\sigma_j]
        \}
\right]_{p}\;
    [\ell_{j+1}\sigma_{j+1}]_{p_{j+1}} \ldots [\ell_k\sigma_k]_{p_k}
\\
&\to_{p}&
t\left[ r\sigma'
        \{
            x_1 {:=} C_1[\ell_1\sigma_1], \dotsc, x_j {:=} C_j[\ell_j\sigma_j]
        \}
\right]_{p}\;
    [\ell_{j+1}\sigma_{j+1}]_{p_{j+1}} \ldots [\ell_k\sigma_k]_{p_k}
&=& \hat{t}
\\
&\parto_{P'}&
t\left[ r\sigma'
        \{
            x_1 {:=} C_1[r_1\sigma_1], \dotsc, x_j {:=} C_j[r_j\sigma_j]
        \}
\right]_{p}\;
    [r_{j+1}\sigma_{j+1}]_{p_{j+1}} \ldots [r_k\sigma_k]_{p_k}
&=& t''
\end{array}
\]

In the second reduction step, the positions of the terms $\ell_i\sigma_i$
in $\hat{t}$ constitute the set $P' \subseteq \Pos(\hat{t})$.
\end{proof}

We are now able to prove our main theorem, showing that context-sensitive
termination implies productivity of the considered proper specification.

\begin{proof}[Proof of Theorem~\ref{thm:CStermination}]
Assume $\Spec = (\Sigma_d,\Sigma_s,\C,\R_d,\R_s)$ is not strongly productive.
Then, a maximal outermost-fair reduction sequence
$\rho \equiv t_0 \to t_1 \to \dots$ exists where for all
$k \in \Nats$, $\rt(t_k) \notin \C$.

This reduction sequence is infinite, since otherwise it would end in a term
$t_m$ for some $m \in \Nats$ with $\rt(t_m) \notin \C$. Then however, according
to Lemma~\ref{lem:ExistRedex}, the term $t_m$ would contain a redex, giving a
contradiction to the sequence being maximal.

The sequence might however perform reductions that are below a variable argument
of a constructor or below a variable in all left-hand sides of a defined symbol.
These reduction steps are not allowed when considering context-sensitive
rewriting with respect to~$\mu$. Such reductions however can be reordered.
First, we observe that there is always a redex which is not blocked, due to 
Lemma~\ref{lem:ExistRedex}, thus there is also an outermost such one. Because
the reduction is outermost-fair, and because reductions below a variable cannot
change the matching of a rule, as shown in Lemma~\ref{lem:SpecialPML}, such
redexes must be contracted an infinite number of times in the infinite reduction
sequence $\rho$. Thus, we can reorder the reduction steps in $\rho$:
If there is a (parallel) reduction below a variable before performing a step
that is allowed by~$\mu$, then we swap these two steps using
Lemma~\ref{lem:SpecialPML}. Repeating this, we get an infinite reduction
sequence $\rho'$ consisting of steps which are not blocked by~$\mu$. Thus, this
is an infinite $\mu$-reduction sequence, showing that $\R_d \cup \R_s$ is not
$\mu$-terminating, which proves the theorem.
\end{proof}

The technique of Theorem~\ref{thm:CStermination}, i.e., proving
$\mu$-termination of the corresponding context-sensitive TRS, is able to prove
strong productivity of Example~\ref{ex:NonSyntax}. By Definition~\ref{def:Mu},
the corresponding replacement map~$\mu$ is defined as
$\mu(\fsym{0}) = \mu(\fsym{1}) = \mu(\fsym{ones}) = \mu(\fsym{finZeroes}) =
\emptyset$ and
$\mu(\fsym{f}) = \mu(\fsym{:}) = \{ 1 \}$, i.e., rewriting is allowed on all
positions except those that are inside a second argument of the
constructor~$:$.
Context-sensitive termination of the TRS together with the above replacement
map~$\mu$ can for example be shown by the tool AProVE~\cite{AProVE06}.
Thus, productivity of that example has been shown according to
Theorem~\ref{thm:CStermination}.
Also, strong productivity of the proper specification consisting of the rules
$\fsym{random} \to \zero : \fsym{random}$,
$\fsym{random} \to \fsym{id}(\one : \fsym{random})$, and $\fsym{id}(xs) \to xs$
can be proven using Theorem~\ref{thm:CStermination} and the tool
AProVE~\cite{AProVE06}, where
$\mu(\zero) = \mu(\one) = \mu(\fsym{random}) = \mu(\fsym{id}) = \emptyset$
and $\mu(:) = \{ 1 \}$
according to Definition~\ref{def:Mu}.
Note that for this example, one could also have used $\mu(\fsym{id}) = \{ 1 \}$,
i.e., here the removal of argument positions in the second item of
Definition~\ref{def:Mu} is irrelevant.

This is not the case in the next example, showing that this improvement,
which was inspired by~\cite{EH11} and blocks more argument positions,
allows to prove productivity of specifications where this would otherwise not be
possible.

\begin{example}
\label{ex:BlockingMoreArgs}
Consider the following proper specification, given by the TRS~$\R_s$:
\[
\begin{array}{r@{\;\;}c@{\;\;}l@{\qquad\qquad}r@{\;\;}c@{\;\;}l}
\fsym{a} &\to& \fsym{f}(\one : \fsym{a}, \fsym{a})
&
\fsym{f}(x : xs, ys) &\to& x : ys
\\&&&
\fsym{f}(\fsym{f}(xs,ys),zs) &\to& \fsym{f}(xs,\fsym{f}(ys,zs))
\end{array}
\]

When defining $\mu(\fsym{1}) = \mu(\fsym{a}) = \emptyset$ and
$\mu(:) = \{ 1 \}$ by the first case of Definition~\ref{def:Mu}, and
defining $\mu(\fsym{f}) = \{ 1,2 \}$ (i.e., not removing any argument
positions, as was done in the orthogonal case in~\cite{ZR10}),
then an infinite $\mu$-reduction exists:
$
    \fsym{a} \;\muto\;
    \fsym{f}(\one : \fsym{a}, \fsym{a}) \;\muto\;
    \fsym{f}(\one : \fsym{a}, \fsym{f}(\one : \fsym{a}, \underline{\fsym{a}}))
    \;\muto\; \ldots
$

\noindent
This reduction can be continued in the above style by reducing the underlined
redex further, which will always create the term~$\fsym{a}$ on an allowed
position of the form $2^n$. However, such positions are not required for any
of the $\fsym{f}$-rules to be applicable; for both rules it holds that all
subterms of left-hand sides that start with the symbol~$\fsym{f}$,
which are the terms $\fsym{f}(x : xs, ys)$, $\fsym{f}(\fsym{f}(xs,ys),zs)$, and
$\fsym{f}(xs,ys)$,
have a variable as second argument. Thus, according to Definition~\ref{def:Mu},
the replacement map~$\mu'$ can be defined to be like $\mu$, except that
$\mu'(\fsym{f}) = \{ 1 \}$.
With this improved replacement map, $\mu'$-termination of the above TRS can for
example be proven by the tool AProVE~\cite{AProVE06},
which implies productivity by Theorem~\ref{thm:CStermination}.
\end{example}

Checking productivity in this way, i.e., by checking context-sensitive
termination, can only prove productivity but not disprove it. This is
illustrated in the next example.

\begin{example}
\label{ex:IncompleteCSterm}
Consider the proper specification with the following rules in $\R_s$:
\[
\begin{array}{rcl@{\qquad\qquad}rcl@{\qquad\qquad}rcl}
\fsym{a} &\to& \fsym{f}(\fsym{a})
&
\fsym{f}(x:xs) &\to& x : \fsym{f}(xs)
&
\fsym{f}(\fsym{f}(xs)) &\to& \one : xs
\end{array}
\]

Starting in the term $\fsym{a}$, we observe that an infinite $\mu$-reduction
starting with $\fsym{a} \to \fsym{f}(\underline{\fsym{a}})$ exists,
which can be continued by reducing the underlined redex repeatedly, since
$\mu(\fsym{f}) = \{ 1 \}$.
Thus, the example is not $\mu$-terminating. However, the specification is
productive, as can be shown by case analysis based on the root symbol of
some arbitrary ground term $t$. In case $\rt(t) = {:}$, then nothing has to be
done, according to Proposition~\ref{prop:NonOrthoProd}.
Otherwise, if $\rt(t) = \fsym{a}$, then any maximal outermost-fair reduction
must start with $t = \fsym{a} \to \fsym{f}(\fsym{a})$, thus we can reduce our
analysis to the final case, where $\rt(t) = \fsym{f}$. In this last case,
$t = \fsym{f}(t')$. Due to the rules for the symbol $\fsym{f}$, we have to
perform a further case analysis based on the root symbol of $t'$.
If $\rt(t') = {:}$, i.e., $t' = u:t''$ for some terms $u$ and $t''$,
then this constructor cannot be reduced further.
Also, $t = \fsym{f}(u:t'')$ is a redex, due to the second rule.
Hence, in any maximal outermost-fair reduction sequence this redex must
eventually be reduced using the second rule, which results in a term with the
constructor~$:$ at the root.
For $\rt(t') = \fsym{a}$ we again must reduce
$t = \fsym{f}(\fsym{a}) \to \fsym{f}(\fsym{f}(\fsym{a}))$.
Finally, in case $\rt(t') = \fsym{f}$, we have two possibilities. The first
one occurs when the term $t'$ is eventually reduced at the root. Since
$\rt(t') = \fsym{f}$, this has to happen with either of the $\fsym{f}$-rules,
creating a constructor~$:$ which, as we already observed, must eventually result
in the term $t$ also being reduced to a term with the constructor~$:$ at the
root. Otherwise, in the second possible scenario, the term $t'$ is never reduced
at the root. Then however, an outermost redex of the shape
$\fsym{f}(\fsym{f}(t''))$ exists in all terms that $t$ can be rewritten to in
this way, thus it has to be reduced eventually with the third rule. This again
creates a term with constructor~$:$ at the root. Combining all these
observations, we see that in every maximal outermost-fair reduction there exists
a term with the constructor~$:$ as root symbol, which proves productivity due to
Proposition~\ref{prop:NonOrthoProd}.
\end{example}

In the remainder of this section we want to illustrate the requirements
of proper specifications in Definition~\ref{def:Spec}, namely that the TRS
$\R_s$ should be left-linear and that structure arguments of constructors in
left-hand sides must not be structure symbols, i.e., they must be variables.
We begin with an example specification that is not left-linear and not
productive, but $\mu$-terminating.

\begin{example}
\label{ex:CexpLeftLin}
We consider the non-proper specification
$\Spec = (\Sigma_d,\Sigma_s,\C,\R_d,\R_s)$ with $\Sigma_d = \R_d = \emptyset$,
$\C = \{ \fsym{a}, \fsym{c} \} \subseteq \Sigma_s =
\{ \fsym{a}, \fsym{b}, \fsym{c}, \fsym{f} \}$,
and the following rules in $\R_s$ which also imply the arities of the symbols:
\[
\begin{array}{rcl@{\qquad}rcl@{\qquad}rcl@{\qquad}rcl}
\fsym{b} &\to& \fsym{a}
&
\fsym{f}(\fsym{a}) &\to& \fsym{a}
&
\fsym{f}(\fsym{c}(x,x)) &\to& \fsym{f}(\fsym{c}(\fsym{a},\fsym{b}))
&
\fsym{f}(\fsym{c}(x,y)) &\to& \fsym{c}(x,y)
\end{array}
\]

The example specification is not productive, as it admits the infinite
outermost-fair reduction sequence
$\fsym{f}(\fsym{c}(\fsym{a},\fsym{a})) \to \fsym{f}(\fsym{c}(\fsym{a},\fsym{b}))
\to \fsym{f}(\fsym{c}(\fsym{a},\fsym{a})) \to \dots$.
However, the TRS is $\mu$-terminating, as shown by the tool
AProVE~\cite{AProVE06}, where $\mu(\fsym{f}) = \{ 1 \}$ and
$\mu(\fsym{a}) = \mu(\fsym{b}) = \mu(\fsym{c}) = \emptyset$. This is the case
because rewriting below the constructor~$\fsym{c}$ is not allowed, thus the
second step of the above reduction sequence is blocked.
The reason why Theorem~\ref{thm:CStermination} fails is the
reordering of reductions, since in this example a reduction of the form
$t \parto_P t' \to_{\ell \to r, p} t''$
(here: $
    \fsym{f}(\fsym{c}(\fsym{a},\fsym{b}))
\parto_{\{1.1\}}
    \fsym{f}(\fsym{c}(\fsym{a},\fsym{a}))
\to_{\fsym{f}(\fsym{c}(x,x)) \to \fsym{f}(\fsym{c}(\fsym{a},\fsym{b})),\epsilon}
    \fsym{f}(\fsym{c}(\fsym{a},\fsym{b}))
$)
does not imply that $t \to_{\ell \to r, p}$
(in the example, $\fsym{f}(\fsym{c}(\fsym{a},\fsym{b}))
\not\to_{\fsym{f}(\fsym{c}(x,x)) \to \fsym{f}(\fsym{c}(\fsym{a},\fsym{b})),
    \epsilon}$),
i.e., Lemma~\ref{lem:SpecialPML} does not hold.
\end{example}

The next example illustrates why non-variable structure arguments of
constructors are not allowed in left-hand sides.

\begin{example}
\label{ex:CexpConsNonVarArg}
Let $\R_s$ contain the following rules:
\[
\begin{array}{rcl@{\qquad}rcl@{\qquad}rcl}
\fsym{ones} &\to& \one : \fsym{ones}
&
\fsym{f}(x : y : xs) &\to& \fsym{f}(y : xs)
&
\fsym{f}(x : xs) &\to& x : xs
\end{array}
\]

Here, we have non-productivity of the corresponding non-proper specification due
to the infinite outermost-fair reduction sequence
$\fsym{f}(\fsym{ones}) \to_1 \fsym{f}(\one : \fsym{ones}) \to_{1.2}
\fsym{f}(\one : \one : \fsym{ones}) \to_{\epsilon} \fsym{f}(\one : \fsym{ones})
\to \ldots$,
however the second step is not allowed when performing context-sensitive
rewriting, since $\mu(\fsym{:}) = \{ 1 \}$.
Using the tool AProVE~\cite{AProVE06}, context-sensitive termination of the
above TRS together with the replacement map~$\mu$ can be shown.

We can however unfold this example (cf.~\cite{EH11,Z09}),
which makes the resulting specification
proper, by introducing a fresh symbol $\fsym{g}$ and replacing the two
rules for $\fsym{f}$ with the following three rules:
\[
\begin{array}{rcl@{\qquad}rcl@{\qquad}rcl}
\fsym{f}(x : xs) &\to& \fsym{g}(x,xs)
&
\fsym{g}(x, y:xs) &\to& \fsym{f}(y:xs)
&
\fsym{g}(x,xs) &\to& x : xs
\end{array}
\]

Then, in the corresponding context-sensitive TRS, we have
$\mu(\fsym{f}) = \mu(\fsym{:}) = \{ 1 \}$, $\mu(\fsym{g}) = \{ 2 \}$, and
$\mu(\fsym{ones}) = \mu(0) = \mu(1) = \emptyset$.
This context-sensitive TRS is not $\mu$-terminating, since it admits the
infinite reduction
$
\fsym{f}(\fsym{ones})
\muto_{1}
\fsym{f}(\one : \fsym{ones})
\muto_{\epsilon}
\fsym{g}(\one, \fsym{ones})
\muto_{2}
\fsym{g}(\one, \one : \fsym{ones})
\muto_{\epsilon}
\fsym{f}(\one : \fsym{ones})
\muto
\dots
$.
\end{example}

It should be noted that the restriction for left-hand sides to only contain
variables in constructor arguments was already made in~\cite{ZR10}. This is the
case because matching constructors nested within constructors would otherwise
invalidate the approach of disallowing rewriting inside structure arguments of
constructors.

\section{Application to Hardware Circuits}
\label{sec:ApplHW}

Proving productivity can be used to verify stabilization of hardware circuits.
In such a circuit, the inputs can be seen as an infinite stream of zeroes and
ones, which in general can occur in any arbitrary sequence. Furthermore, a
circuit contains a number of internal signals, which also carry different
Boolean values over time.

To store a value over time, feedback loops are used. In such a loop, a value
that is computed from some logic function is also used as an input to that
function. Thus, it is desired that such values stabilize, instead of oscillating
infinitely.

To check this, productivity analysis can be used. We will illustrate this by
means of an example, that will be considered throughout the rest of this
section.

\begin{figure}
\begin{center}
\begin{tikzpicture}[
]
\node[trapezium, draw,
    trapezium left angle=70, trapezium right angle=70,
    rotate=-90,
    inner sep=.5cm
    ]
            (muxi) at (-2,0) {};
\node[circle,draw,right=-.4pt
    ]
        (minvi) at (muxi.top side) {};

\node[signal,draw,anchor=east]
        (d) at (muxi.south west) {\texttt{D}};
\node[signal,draw,anchor=east]
        (si) at (muxi.south east) {\texttt{SI}};
\node[signal,draw]
        (se) at (-3,-1.5) {\texttt{SE}};
\node[signal,draw]
        (ck) at (-3,-2.5) {\texttt{CK}};

\node[shape=not gate US,draw,
    logic gate inverted radius=.15cm
    ]
        (ckn) at (-1.5,-2.5) {};

\node[shape=not gate US,draw,
    logic gate inverted radius=.15cm
    ]
        (cknn) at (1,-2.5) {};

\draw[] (se.east) -| (muxi.east);
\draw[] (ck.east) -- (ckn.input);

\draw[] (ckn.output) -- (cknn.input);

\node[trapezium, draw, 
    trapezium left angle=70, trapezium right angle=70,
    rotate=-90,
    inner sep=.5cm
    ]
            (mux1) at (0,0) {};
\node[circle,draw,right=-.4pt
    ]
        (minv1) at (mux1.top side) {};

\node[shape=not gate US,draw,rotate=180,
    logic gate inverted radius=.15cm
    ]
        (inv1) at (0.25,2) {};

\draw[] (minvi.east) -- (-1,0) |- (mux1.south east)
        node[below,pos=0.5] {\texttt{next}};
\draw[] (ckn.output) -| (mux1.east);

\draw[] (minv1.east) -- (1,0) |- (inv1.input);
\draw[] (inv1.output) -- (-1,2) |- (mux1.south west);

\node[trapezium, draw, 
    trapezium left angle=70, trapezium right angle=70,
    rotate=-90,
    inner sep=.5cm
    ]
            (mux2) at (2.5,0) {};
\node[circle,draw,right=-.4pt
    ]
        (minv2) at (mux2.top side) {};

\node[shape=not gate US,draw,rotate=180,
    logic gate inverted radius=.15cm
    ]
        (inv2) at (2.75,2) {};

\draw[] (minv1.east) -- (1.5,0) |- (mux2.south east)
        node[below,pos=0.5] {\texttt{n1}};
\draw[] (cknn.output) -| (mux2.east);

\draw[] (minv2.east) -- (3.5,0) |- (inv2.input)
        node[right,pos=0.25] {\texttt{n2}};
\draw[] (inv2.output) -- (1.5,2) |- (mux2.south west);

\node[shape=not gate US,draw,
    logic gate inverted radius=.15cm
    ]
        (invq) at (4,0) {};

\node[shape=not gate US,draw,
    logic gate inverted radius=.15cm
    ]
        (invqn) at (5.5,2) {};

\draw[] (minv2.east) -- (invq.input);
\draw[] (invq.output) -- (5,0) |- (invqn.input);

\node[signal,draw] (q) at (7,0) {\texttt{Q}};
\node[signal,draw] (qn) at (7,2) {\texttt{QN}};

\draw[] (invq.output) -- (q.west);
\draw[] (invqn.output) -- (qn.west);

\end{tikzpicture}
\end{center}
\caption{Example hardware circuit}
\label{fig:Circuit}
\end{figure}
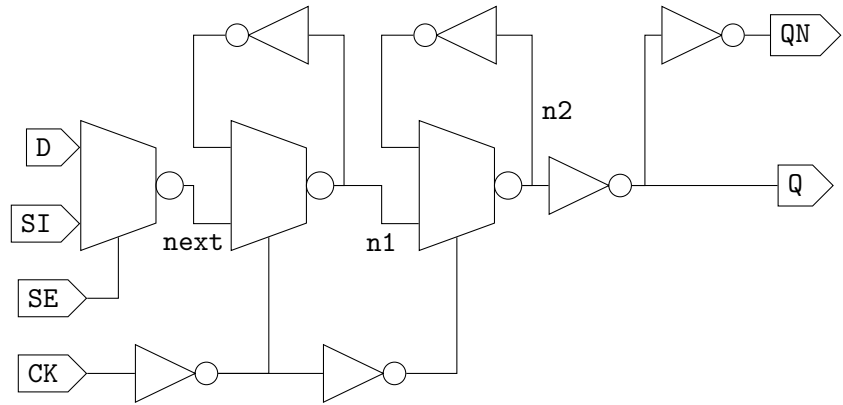

Consider the circuit shown in Figure~\ref{fig:Circuit},
which was constructed from the transistor netlist of the cell \texttt{SDFF\_X1}
in the Nangate Open Cell Library~\cite{OpenCellLibrary08_10_SP1}
and which implements a scanable D~flip-flop.
This circuit first selects, based on the value of the input
\texttt{SE} (scan enable), either the negation of the data input \texttt{D}
(in case \texttt{SE}=0)
or the negation of the scan data input \texttt{SI} (in case \texttt{SE}=1).
This value, called \texttt{next} in Figure~\ref{fig:Circuit}, is then fed into
another multiplexer (mux), for which a feedback loop exists.
This mux is controlled by the negation of the clock input \texttt{CK}. If the
clock is~0 then the negated value of \texttt{next} is forwarded to the output
\texttt{n1}, otherwise the stored value of \texttt{n1} is kept.
Similarly, \texttt{n2} implements such a latch structure, however this time the
latch forwards the negation of the \texttt{n1} input in case \texttt{CK} is 1,
and it keeps its value when \texttt{CK} is 0.
The outputs \texttt{Q} and \texttt{QN} are computed from this stored value
\texttt{n2}.

Note that a lot of the negations are only contained to refresh the signals,
otherwise a high voltage value might decay and not be detected properly anymore.

From the example circuit, we create a proper specification, where the data
symbols consist of the two Boolean values~$\zero$ and~$\one$ and the symbol
$\fsym{not}$ used for negating:
\[
\begin{array}{rcl@{\qquad\qquad}rcl}
\fsym{not}(\zero) &\to& \one
&
\fsym{not}(\one) &\to& \zero
\end{array}
\]

The structures we are interested in are infinite streams containing Boolean
values, thus the set of constructors is $\C = \{ {:} \}$. The structure
TRS $\R_s$ is shown in Figure~\ref{fig:HWrules}.

\begin{figure}[t]
\centering
\begin{align*}
  \fsym{rand} &\to \zero : \fsym{rand}
\\
  \fsym{rand} &\to \one : \fsym{rand}
\\
  \fsym{next}(\zero : ses,d : ds,si : sis) &\to
        \fsym{not}(d) : \fsym{next}(ses,ds,sis)
\\
  \fsym{next}(\one : ses,d : ds,si : sis) &\to
        \fsym{not}(si) : \fsym{next}(ses,ds,sis)
\\
\\
  \fsym{n1}(\zero : cks,\textit{nextv} : \textit{nexts},\textit{n1l})
&\to
    \fsym{not}(\textit{nextv}) : \fsym{n1}(cks,\textit{nexts},
                                                \fsym{not}(\textit{nextv}))
\\
  \fsym{n1}(\one : cks,\textit{nextv} : \textit{nexts},\textit{n1l})
&\to
            \fsym{n1}'(cks,\textit{nexts},\textit{n1l},
                    \fsym{not}(\fsym{not}(\textit{n1l})))
\\
  \fsym{n1}'(cks,\textit{nexts},\zero,\zero)
&\to
    \zero : \fsym{n1}(cks,\textit{nexts},\zero)
\\
  \fsym{n1}'(cks,\textit{nexts},\one,\one)
&\to
    \one : \fsym{n1}(cks,\textit{nexts},\one)
\\
  \fsym{n1}'(cks,\textit{nexts},\zero,\one)
&\to
        \fsym{n1}'(cks,\textit{nexts},\one,\fsym{not}(\fsym{not}(\one)))
\\
  \fsym{n1}'(cks,\textit{nexts},\one,\zero)
&\to
        \fsym{n1}'(cks,\textit{nexts},\zero,\fsym{not}(\fsym{not}(\zero)))
\\
\\
\fsym{n2}(\zero : cks,\textit{n1v} : \textit{n1s},\textit{n2l})
&\to
      \fsym{n2}'(cks,\textit{n1s},\textit{n2l},
            \fsym{not}(\fsym{not}(\textit{n2l})))
\\
\fsym{n2}(\one : cks,\textit{n1v} : \textit{n1s},\textit{n2l})
&\to
      \fsym{not}(n1v) :\fsym{n2}(cks,\textit{n1s},\fsym{not}(\textit{n1v})))
\\
\fsym{n2}'(cks,\textit{n1s},\zero,\zero) &\to
            \zero : \fsym{n2}(cks,\textit{n1s},\zero)
\\
\fsym{n2}'(cks,\textit{n1s},\one,\one) &\to 
            \one : \fsym{n2}(cks,\textit{n1s},\one)
\\
\fsym{n2}'(cks,\textit{n1s},\zero,\one) &\to
      \fsym{n2}'(cks,\textit{n1s},\one,\fsym{not}(\fsym{not}(\one)))
\\
\fsym{n2}'(cks,\textit{n1s},\one,\zero) &\to
      \fsym{n2}'(cks,\textit{n1s},\zero,\fsym{not}(\fsym{not}(\zero)))
\\
\\
\fsym{q}(\textit{n2v} : \textit{n2s}) &\to
    \fsym{not}(\textit{n2v}) : \fsym{q}(\textit{n2s})
\\
\\
\fsym{qn}(\textit{qv} : \textit{qs}) &\to
    \fsym{not}(\textit{qv}) : \fsym{qn}(\textit{qs})
\end{align*}
\caption{Structure TRS $\R_s$ for the circuit shown in Figure~\ref{fig:Circuit}}
\label{fig:HWrules}
\end{figure}

It should be remarked that in the shown rules, some simplifications regarding
the clock input \texttt{CK} have been made. The inverters for the clock have
been removed, and the two muxes that output the signals \texttt{n1} and
\texttt{n2} are provided with decoupled clock values.

The defined function symbols $\fsym{next}$, $\fsym{n1}$, $\fsym{n2}$,
$\fsym{q}$, and $\fsym{qn}$ reflect the wires and output signals with the
corresponding name in Figure~\ref{fig:Circuit}.
The constant~$\fsym{rand}$ is added to abstract the values of the inputs. It
provides a random stream of Boolean values, thus it is able to represent any
sequence of input values provided to the circuit.
The rules of the symbol $\fsym{next}$ implement the mux selecting either the
next data input value $d$ in case the next scan enable input value~$se$ is~$0$,
or the next scan input value $si$ in case~$se$ is~$1$.

The output of $\fsym{n1}$ is also computed by a mux, however, here the previous
output value has to be considered due to the feedback loop. We break the cycle
by introducing a new parameter~$\textit{n1l}$ that stores the previously output
value.
Then, the next value of the stream at $\fsym{n1}$ is computed
from the next value of the clock~$ck$, the input stream
$\textit{nextv}:\textit{nexts}$ coming from the previously described
multiplexer, and from the previous output value~$\textit{n1l}$.
If the clock~$ck$ is~$\zero$, then the latch simply outputs the negated value of
$\textit{nextv}$ and continues on the remaining streams, setting the
parameter~$\textit{n1l}$ to this value to remember it.
Otherwise, if~$ck$ is~$\one$, then the feedback loop is active and has to be
evaluated until it stabilizes. This is done by the function $\fsym{n1'}$.
It has as arguments the remaining input stream of the clock, the remaining input
stream of the scan multiplexer, and the previous output value and the newly
computed output value. If both of these values are the same, then the value of
the wire $\fsym{n1}$ has stabilized and hence can be output. The tail of the
output stream is computed by again calling the function $\fsym{n1}$ with the
remaining streams for the clock and the scan multiplexer.
Otherwise, the new output value (the last argument of $\fsym{n1'}$) differs from
the old output value (the penultimate argument of $\fsym{n1'}$). In that case,
the new output value becomes the old output value and the new output is
recomputed. This is repeated until eventually the output value stabilizes, or it
will oscillate and never produce a stable output.

Similar to the function~$\fsym{n1}$, the function~$\fsym{n2}$ computes stable
values for the corresponding wire in Figure~\ref{fig:Circuit}. Again, the
parameter~$\textit{n2l}$ is added to store a previously output value, and the
auxiliary function~$\fsym{n2'}$ is used to compute a stable value for the
feedback loop. The only difference to the function~$\fsym{n1}$ is that the cases
of the clock are inverted, due to the additional inverter in
Figure~\ref{fig:Circuit} that feeds the select input of the multiplexer that
computes~$\fsym{n2}$.
Finally, the functions~$\fsym{q}$ and~$\fsym{qn}$ implement the two inverters
that feed the corresponding output signals in Figure~\ref{fig:Circuit}.

The above specification is productive, since the TRS $\R_d \cup \R_s$ can be
proven context-sensitive terminating, for example by the tool
AProVE~\cite{AProVE06}.
Hence, according to Theorem~\ref{thm:CStermination}, the specification is
productive, meaning that every ground term of sort~$s$ rewrites to a constructor
term. This especially holds for the ground terms
$t_{\fsym{q}} = \fsym{q}(t_{\fsym{n2}})$
and
$t_{\fsym{qn}} = \fsym{qn}(t_{\fsym{q}})$,
where
$t_{\fsym{n2}} = \fsym{n2}(
    \fsym{rand},
    \fsym{n1}(
        \fsym{rand},
        \fsym{nexts}(
            \fsym{rand},\fsym{rand},\fsym{rand}
        ),
        \textit{n1l}
    ),
    \textit{n2l}
)$
and the variables $\textit{n1l}$ and $\textit{n2l}$ are instantiated with all
possible combinations of~$\zero$ and~$\one$. Thus, the circuit produces an
infinite stream of stable output values, regardless of its initial state and
input streams, and does not oscillate infinitely long. This illustrates that
productivity analysis can be used to prove stabilization of digital circuits
with arbitrary input sequences, when encoding them as non-orthogonal
proper specifications.

\section{Conclusions and Future Work}
\label{sec:Conclusions}

We have presented a generalization of the productivity checking techniques
in~\cite{ZR10} (including the improvements of~\cite{EH11})
to non-orthogonal specifications, which are able to represent non-deterministic
systems. These naturally arise for example when abstracting away certain details
of an implementation, such as the concrete sequence of input values. This was
used to verify stabilization of hardware descriptions whose environment is left
unspecified, as was demonstrated in Section~\ref{sec:ApplHW}.

Our setting still imposes certain restrictions on the specifications that can be
treated. The most severe restriction is the requirement of left-linear rules in
the structure TRS $\R_s$.
Dropping this requirement however would make Theorem~\ref{thm:CStermination}
unsound.
Similarly, also the requirement that structure arguments of constructors must be
variables cannot be dropped without losing soundness of
Theorem~\ref{thm:CStermination}. This requirement however is not that severe in
practice, since many specifications can be unfolded by introducing fresh
symbols, as was presented in~\cite{EH11,Z09}.

In the future, it would be interesting to investigate whether transformations of
non-orthogonal proper specifications, similar to those in~\cite{ZR10}, can be
defined. It is clear that rewriting of right-hand sides for example is not
productivity-preserving for non-orthogonal specifications, since it only
considers one possible reduction.
However, it would be interesting to investigate whether for example narrowing of
right-hand sides is productivity preserving, as it considers all possible
reductions.

\noindent
\paragraph*{Acknowledgment.} The author would like to thank the anonymous
reviewers for their valuable comments and suggestions that helped to improve
the paper.

\bibliographystyle{eptcs}
\bibliography{ref}

\end{document}